\documentclass[11pt]{article}
\usepackage{amssymb,amsmath,graphicx,amsthm}
\usepackage{cite}
\usepackage{xcolor}
\usepackage{hyperref,cleveref}
\usepackage{algorithm}
\usepackage{algpseudocode}
\usepackage[font=small]{caption,subcaption}
\usepackage[paper=letterpaper,margin=1in]{geometry}

\newtheorem{theorem}{Theorem}

\newtheorem{corollary}[theorem]{Corollary}
\newtheorem{lemma}[theorem]{Lemma}

\newtheorem{definition}{Definition}

\newcommand{\R}{\mathbb{R}}
\newcommand{\C}{\mathbb{C}}

\newcommand{\toep}{\mathrm{Toep}}
\newcommand{\norm}[1]{\|#1\|}

\DeclareMathOperator{\rank}{rank}
\DeclareMathOperator{\diag}{diag}

\title{Low-Rank Toeplitz Matrix Estimation \protect\\  via Random Ultra-Sparse Rulers}
\author{\and \and \and \and \and Hannah Lawrence \\ Flatiron Institute\\ \texttt{hllawrence@gmail.com} \and Jerry Li \\ Microsoft Research \\ \texttt{jerrl@microsoft.com} \and \and \and \and \and Cameron Musco \\ UMass Amherst \\ \texttt{cmusco@cs.umass.edu} \and Christopher Musco \\ New York University \\ \texttt{cmusco@nyu.edu}}

  \usepackage{nth}
  \usepackage{intcalc}

\begin{document}
%
\maketitle
\begin{abstract}
We study how to estimate a nearly low-rank Toeplitz covariance matrix $T$ from compressed measurements. 
Recent work of Qiao and Pal addresses this problem by combining sparse rulers (sparse linear arrays) with frequency finding (sparse Fourier transform) algorithms applied to the Vandermonde decomposition of $T$. Analytical bounds on the sample complexity are shown, under the assumption of sufficiently large gaps between the frequencies in this decomposition.


In this work, we introduce \emph{random ultra-sparse rulers} and propose an improved approach based on these objects. Our random rulers effectively apply a random permutation to the frequencies in $T$'s Vandermonde decomposition, letting us avoid frequency gap assumptions and leading to improved sample complexity bounds. In the special case when $T$ is circulant, we theoretically analyze the performance of our method when combined with 
sparse Fourier transform algorithms based on random hashing. We also show experimentally that our ultra-sparse rulers give significantly more robust and sample efficient estimation then baseline methods. 
\end{abstract}
%
%
\clearpage

\section{Introduction}
\label{sec:intro}

We study the problem of estimating the $d \times d$ covariance matrix  $T \in \R^{d \times d}$ of a distribution $\mathcal{D}$ over $d$-dimensional vectors given independent samples $x^{(1)},x^{(2)},\ldots, x^{(n)} \in \R^d$ drawn from $\mathcal{D}$. In particular, we focus on the case when the covariance matrix $T$ is Toeplitz, which arises when the vectors are wide-sense stationary: the covariance $t_{|j-k|}$ between the $j^{th}$ and $k^{th}$ entries only depends on the distance $|j -k|$. We let $t_s$ denote the covariance at distance $s$ for $s \in \{0,\ldots,d-1\}$.

Toeplitz covariance estimation arises in a range of applications, including direction of arrival (DOA) estimation  \cite{KrimViberg:1996,wax1997joint,BogaleLe:2016}, spectrum-sensing for cognitive radio \cite{MaLiJuang:2009,CohenTsiperEldar:2018}, medical and radar imaging \cite{SnyderOSullivanMiller:1989, fuhrmann1991application,BrookesVrbaRobinson:2008}, \cite{CohenEldar:2018,RufSwiftTanner:1988,AslMahloojifar:2012} 
and Gaussian process regression (kriging) and kernel machine learning \cite{dietrich1991estimation,wilson2015kernel}.
We focus on estimation methods  with \emph{low sample complexity}, can be measured in two ways \cite{toeplitz}:

\paragraph{Entry Sample Complexity.}  How many entries of each sample $x^{(i)} \sim \mathcal{D}$ must be read? Minimizing entry sample complexity typically corresponds to minimizing sensor cost, as, in many applications, each entry of $x^{(i)}$ is measured with a different sensor in a spatial grid. We consider algorithms where the same entries are read in each  $x^{(i)}$ (i.e., the active sensors remain fixed).

\paragraph{Vector Sample Complexity.}  How many $d$-dimensional samples $x^{(i)}$ must be drawn from $\mathcal{D}$? Vector sample complexity corresponds to minimizing acquisition time or measurement cost and  is the classic notion of sample complexity in statistics and machine learning.

Typically  there is a trade-off between these two measures. In this work, we seek to \emph{minimize entry sample complexity}, while keeping vector sample complexity reasonably low.

\subsection{Sparse Ruler Based Sampling}
Our work centers on the powerful idea of sparse rulers (also known as sparse linear arrays), which let one perform covariance estimation with significantly reduced entry  sample complexity. A sparse ruler is a subset of indices $R \subseteq \{1,\ldots, d\}$, such that for every  distance $s \in \{0,1,\ldots, d-1\}$, there is some pair $i,j \in R$ with distance $| i - j| =  s$. The set of distances measured by $R$ is $R$'s \emph{difference coarray} or \emph{difference set} \cite{ErdosGal:1948,Leech:1956, Wichmann:1963,Moffet:1968,PillaiBar-NessHaber:1985}.
 It is clear that to represent $d$ distances, we must have $|R| \ge \sqrt{d}$ so that ${R \choose 2} \ge d$ 
and it is well known that for any $d$, there exists a sparse ruler matching this optimal size up to constants. A large body of work has studied the design of  sparse rulers under various additional objectives \cite{caratelli2011novel,qin2015generalized,cohen2019sparse}. We note that in some cases, which will arise later in this work, we may allow $R$ to be any set of integers, including those outside $\{1,\ldots, d\}$.

Sparse rulers have received significant attention in covariance estimation applications \cite{LexaDaviesThompson:2011,ArianandaLeus:2012,romero2016compressive,WuZhuYan:2017,toeplitz}.
Given a sample $x \sim \mathcal{D}$ with Toeplitz covariance matrix $T$, if we read the $|R|$ entries of $x$ corresponding to indices in a ruler $R$, we obtain an estimate of the covariance $t_s$ \emph{at every distance $s$}. So in principle, with enough samples, $x^{(1)},\ldots, x^{(n)} \sim \mathcal{D}$ we can accurately estimate $T$ while measuring just $|R| = O(\sqrt{d})$ entries in each sample (i.e., with $O(\sqrt{d})$ entry sample complexity). In fact, recent work has shown that, with sparse ruler measurements, $\tilde O(d/\epsilon^2)$ vector samples suffice to recover any Toeplitz matrix to accuracy $\epsilon$ in the spectral norm\cite{toeplitz}.

\subsection{Improved Bounds for Low-Rank Matrices}\label{sec:piya}

For general Toeplitz covariance matrices it is impossible to improve on the entry sample complexity achieved by sparse rulers: without reading at least $O(\sqrt{d})$ entries, we can never estimate the covariance at some distances.
However 
in many applications, such as DOA estimation, when the number of sources is smaller than the number of sensors, the Toeplitz covariance matrix of the received signal snapshots is \emph{low-rank}, or close to low-rank. This \emph{additional structure} can be leveraged to recover $T$ with a smaller subset of its entries \cite{abramovich1996positive,ChenChiGoldsmith:2015}.
Recent work of Qiao and Pal \cite{qiao2017gridless} shows that, if $T$ is approximately rank $k$ for any $k < d$, an entry sample complexity of just $O(\sqrt{k})$ can be achieved using sparse rulers. 
The high-level idea is easily understood: if $T$ is exactly rank-$k$, then it can be decomposed uniquely using the Carath\'{e}odory-Fej\'{e}r-Pisarenko decomposition (the Vandermonde decomposition) \cite{caratheodory1911zusammenhang} 
 as $T = F_T D F_T^*$, where $D \in \R^{k \times k}$ is a diagonal matrix and $F_T \in \C^{d \times k}$ is a Fourier matrix, with $F_T({m,\ell}) = e^{2 \pi i f_\ell \cdot (m-1)}$ for some set of frequencies $f_1,\ldots, f_k \in [0,1]$.

We can see immediately that the top left $k + 1 \times k+1$ principal submatrix of $T$, denoted $T_{k+1,k+1}$ 
(which is also Toeplitz, positive semidefinite, and rank $k$) 
admits a Vandermonde decomposition with  the same frequencies -- obtained by simply restricting $F$ to its first $k+1$ rows. Further, it can be shown that this decomposition is unique. Thus, we can recover the frequencies $f_1,\ldots, f_k$  and their weights $D$ just from a decomposition of $T_{k+1,k+1}$. Thus, from this small submatrix, we can recover  all of $T$!

With this observation in hand, Qiao and Pal  apply  sparse ruler methods to $T_{k+1,k+1}$ to obtain entry sample complexity just $O(\sqrt{k})$. The key difficulty is that the Vandermonde decomposition is notoriously unstable: noise in approximating $T_{k+1,k+1}$ and any deviation of $T$ from being exactly rank-$k$ (i.e., truly having just $k$ frequencies in its Vandermonde decomposition) can entirely change the frequency content of this decomposition. Nevertheless, Qiao and Pal prove a bound on reconstruction error, under the assumption that $f_1,\ldots, f_k$ have spacing at least $\Theta(1/k)$ and that the underlying MUSIC frequency-finding routine \cite{Schmidt:1981,Schmidt:1986} is exact. They give a vector sample complexity bound of roughly $O \left ({d^4}/{k^2 \epsilon^2} \right )$ to approximate all entries of $T$ up to error $\epsilon \cdot  t_0$, where $t_0$ is covariance at distance $0$ (and therefore the largest entry of $T$ since it is positive semidefinite).

\subsection{Our Contributions}\label{sec:contributions}
We propose the idea of \emph{random ultra-sparse rulers} to avoid the frequency gap assumption of Qiao and Pal, while simultaneously giving much lower  vector sample complexity with similar entry sample complexity. In the special case when $T$ is circulant (corresponding to frequencies in its Vandermonde decomposition being `on-grid' multiples of $1/d$), we prove a sample complexity bound that depends only \emph{logarithmically on the ambient dimension $d$}, and polynomially on the rank $k$ and error parameter $\epsilon$. Broadly, our \emph{random ultra-sparse rulers} open the door to achieving low entrywise sample complexity for circulant Toeplitz covariance estimation via a wider class of randomized sparse FFTs, providing more robust frequency recovery than deterministic techniques.


\begin{theorem}[Circulant Covariance Estimation]\label{thm:main}
Algorithm \ref{alg:alg} takes $\tilde O(k/\epsilon^2)$\footnote{$\tilde O(\cdot)$ hides log factors in the input parameters. For more precise bounds see Section \ref{sec:analysis}.} independent samples from any  sub-Gaussian distribution $\mathcal{D}$ on $\R^d$ with circulant covariance matrix $T$. The algorithm reads $\tilde O(\sqrt{k})$ entries from each sample and returns with probability at least $2/3$, $\tilde T \in \R^{d \times d}$ satisfying:
\begin{align*}
\norm{T - \tilde T}_F \le \epsilon \norm{T}_F + 2 \min_{\text{rank-$k$}\ B} \norm{T - B}_F.
\end{align*}
Throughout, $\norm{\cdot }_F$ denotes the Frobenius norm.
\end{theorem}

\paragraph{Random Hashing for Ruler Design.}
Algorithm \ref{alg:alg} (Section \ref{sec:analysis}) is inspired by work on random hashing based sparse Fourier transform methods
\cite{sft,gilbert2014recent}. The idea is to \emph{transform $T$} in way that is equivalent to applying a random hash function ${h}: [0,1] \rightarrow [0,1]$ to the frequencies $f_1,\ldots, f_d$ in $T$'s Vandermonde decomposition. When $T$ is nearly rank-$k$, there may be up to $d$ such frequencies, but only $k$ will significantly contribute to the decomposition. After hashing, we expect the $k$ dominant frequencies to be well separated (without small gaps), and thus recoverable via a frequency finding approach like that proposed by Qiao and Pal. Even if some small gaps remain, by applying repeated random hash functions we can eventually recover all $k$ significant frequencies.

As utilized in sparse Fourier transform methods,
%
 when all frequencies $f_1,\ldots, f_d$ are on-grid integer multiples of $1/d$ (i.e., $T$ is circulant), it suffices to chose $h$ from the family of random hash functions ${h}_{a,b}(x) = a (x-b) \mod d $, where $a, b$ are randomly chosen integers \cite{gilbert2005improved,gilbert2014recent}. $h_{a,b}(x)$ is applied to $x = fd$ when $f$ is an on-grid frequency in $\{0,1/d, \ldots, (d-1)/d\}$, and the hashed frequency is taken as $\frac{h_{a,b}(x)}{d} \in [0,1]$.  
 
 Critically, the random hash function $ {h}_{a,b}$ in frequency domain can be implemented simply via a transformation to $T$. 
 For a random integer $c$, let ${g}_{a,c}(x)  = a (t- c) \mod d$.  If $a$ is coprime to $d$, $ {g}_{a,c}(\cdot)$ is a permutation of $\{0,\ldots, d-1 \}$. Let $T_{a,b,c}$ be a transformed covariance matrix obtained by permuting $T$'s rows and columns with ${g}_{a,c}(x)$ and multiplying the $j,k$ entry by $e^{\frac{2 \pi i ab |j-k|}{d}}$. Let $\tilde{f}_1, \ldots, \tilde{f_d}$ and $\tilde{D}$ denote the frequencies and diagonal matrix in $T_{a,b,c}$'s Vandermonde decomposition. 
 One can  check that $T$'s Vandermonde decomposition can be obtained by setting $f_j = \frac{h_{a,b}^{-1}(\tilde{f_j}d)}{d}$ and $D = P\tilde{D}$ where $P$ is diagonal with $j^{th}$ entry $e^{2 \pi i ac f_j}$. 
 
 
Accordingly, estimating $T$ reduces to estimating $T_{a,b,c}$, which we will do by estimating $T_{a,b,c}$'s top $O(k) \times O(k)$ submatrix and applying a strategy similar to \cite{qiao2017gridless}.
 %
%
   Naively, if the permutation $ {g}_{a,c}(\cdot)$ were truly random, it would destroy the possibility of using a sparse ruler to measure this top submatrix of $T_{a,b,c}$: a general ruler construction is not known for a ruler with the \emph{arbitrary} difference set a random permutation would require. However, by leveraging $ {g}_{a,c}(\cdot)$'s simple structure, we show that we can still construct a ruler to read this submatrix. 
   The ruler is 1) \emph{random}: based on randomly chosen $a,c$ and 2) \emph{ultra-sparse}: measuring the covariance at $O(k)$ random distances using just $O(\sqrt{k})$ entry sample complexity.

\section{Random Ultra-Sparse Rulers}
\label{sec:ultra}

 We start with a simple random ultra-sparse ruler construction that will suffice for circulant matrix estimation.

\begin{definition}[Random Ultra-Sparse Ruler -- Type 1]\label{def:rus}
For any $d$ and $k \le d$, 
let $a,c\in \mathbb Z$ be chosen randomly such that $a$ is coprime to $d$. Let $Q_{a,c} = \{q_1,q_2,\ldots, q_m\}$ be any ruler for the distance set $\{a(0-c), a(1-c), \ldots, a(k-c)\}$ and let $R_{a,c} = \{r_1,\ldots, r_m\}$ where $r_i = q_i \mod d$. We call $R_{a,c}$ a \emph{random ultra-sparse ruler}.
\end{definition}
\begin{lemma}\label{clm:cyclic}
Let $R_{a,c}$ be constructed from any valid $Q_{a,c}$ as in Definition \ref{def:rus} and $ {g}_{a,c}(x) = a(x-c) \mod d$ be the random permutation corresponding to $a,c$. Then the following hold:
\begin{enumerate}
\item $R_{a,c}$ is a \emph{cyclic ruler} for $\{ {g}_{a,c}(0), {g}_{a,c}(1),\ldots,  {g}_{a,c}(k)\}$. I.e., for any $s \in \{0,\ldots,k\}$ there are $r_i,r_j \in R_{a,c}$ with either $r_i - r_j = g_{a,c}(s)$ or $r_i - r_j = d - g_{a,c}(s)$.
\item There exists $Q_{a,c}$, a ruler for the difference set $\{a(0-c), a(1-c), \ldots, a(k-c)\}$, with $|Q_{a,c}| = O(\sqrt{k})$. Correspondingly, $ |R_{a,c}| = O(\sqrt{k})$.
\end{enumerate}
\end{lemma}
\begin{proof}
Since $Q_{a,c}$ is a ruler for the distance set $\{a(0-c), a(1-c), \ldots, a(k-c) \}$, for any $s \in \{0,\ldots, k\}$ there is some pair $q_i,q_j$ with $q_i - q_j = a(s - c)$ Thus $q_i - q_j \equiv a(s - c) \ \mod d$ and so $r_i - r_j  \equiv a(s - c) \ \mod d$ and so $r_i - r_j \equiv g_{a,c}(s)\ \mod d$. 
Since $r_i,r_j$, and $g_{a,c}(s)$ are in $\{0,\ldots,d-1\}$, this equivalence can only hold if $r_i - r_j =  g_{a,c}(s)$ or $r_i - r_j = g_{a,c}(s) -d$ and so $r_j - r_i = d - g_{a,c}(s)$. This completes the first  claim.

For the second claim, set $Q_{a,c} = \{0,a,\ldots, \lceil \sqrt{k}\rceil \cdot a\} \bigcup \{a \lceil \sqrt{k} \rceil - ac,2a \lceil \sqrt{k} \rceil -ac,\ldots, \lceil \sqrt{k}\rceil \cdot a \lceil \sqrt{k} \rceil-ac\}$, as shown in Figure \ref{fig:ruler}. 
We can see that $|Q_{a,c}| = |R_{a,c}| \leq 2\lceil \sqrt{k}\rceil +1$. Just considering distances between the first and second halves of the ruler, $Q_{a,c}$'s difference coarray includes $as - ac$ for all nonnegative $s \le \lceil \sqrt{k}\rceil^2$. So $Q_{a,c}$ is an ultra-sparse ruler for the distance set $\{a(0-c), a(1-c), \ldots, a(k-c)\}$, as required.
\end{proof}

\begin{figure}[h]
\centering

\includegraphics[width=0.5\linewidth]{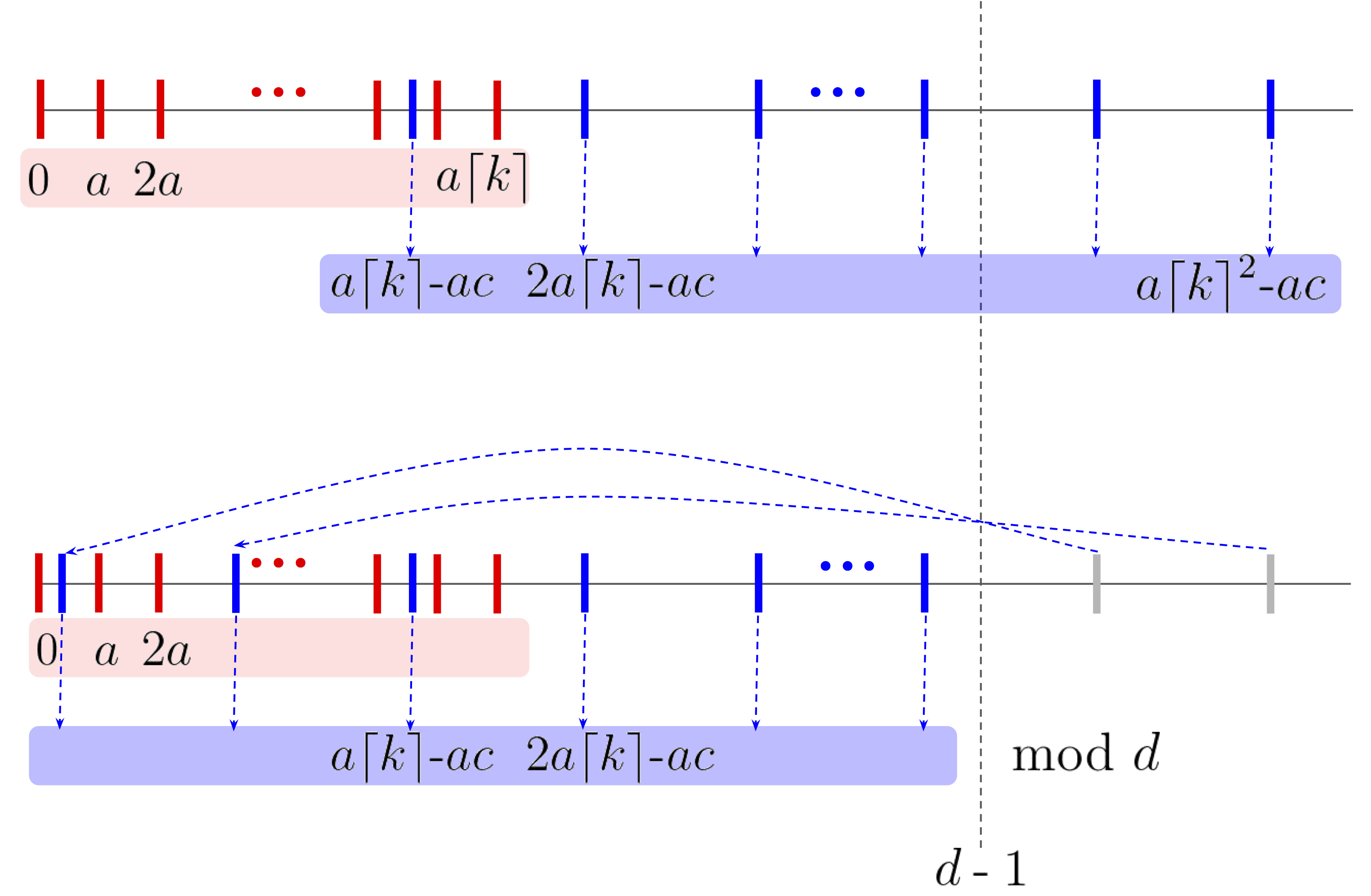}
\caption{Illustration of rulers in Lemma \ref{clm:cyclic}. The top line shows $Q_{a,c}$, with the first set $\{0,a,\ldots, \lceil \sqrt{k}\rceil \cdot a\}$ shown in red and the latter $\{a \lceil \sqrt{k} \rceil - ac,2a \lceil \sqrt{k} \rceil -ac,\ldots, \lceil \sqrt{k}\rceil \cdot a \lceil \sqrt{k} \rceil-ac\}$ in blue. Considering pairwise distances between red and blue markers demonstrates that the difference set of $Q_{a,c}$ is as claimed. Note that the elements of $Q_{a,c}$ may be far greater than $d$, and they may even be negative (for simplicity, in this illustration we assume $\lceil k \rceil > c$). The bottom line visualizes $R_{a,c}$, which is $Q_{a,c}$ ``wrapped around" $\mod d$.}
\label{fig:ruler}

\end{figure}

Note that in a circulant matrix $T$, we have $t_i = t_{d - i}$. Thus a cyclic ruler of the form guaranteed by Lemma \ref{clm:cyclic} suffices to measure the covariance at the full set of random distances $\{{g}_a(0), \ldots,  {g}_a(k)\}$. In Section 3, we will show how this precise structure of difference set is just what's needed by an efficient existing sparse FFT for on-grid frequencies, and we derive corresponding error guarantees for circulant covariance estimation. However, for general Toeplitz matrices (i.e., not cyclically symmetric), we require a true ruler. In this case, we can restrict the range of $a$ to prevent wrap around. For simplicity, in the following definition we also do not implement a random shift $c$.
\begin{definition}[Random Ultra-Sparse Ruler -- Type 2]\label{def:rus2}
For any dimension $d$ and $k \le d$, 
let $a\in \mathbb Z$ be chosen randomly such that $a$ is coprime to $d$ and $a \le bd/k$ for some $b \leq 1$. Let $R_a = \{r_1,r_2,\ldots, r_m\}$ be any ruler for $\{0, a, 2a , \ldots, k a\}$.
\end{definition}
Again, it is clear that we can find $R_a$ with $m = O(\sqrt{k})$. While in this manuscript we do not fully cover how to recover a non-circulant $T$ from a Type 2 ultra-sparse ruler, we give a short sketch here. If we set $k' = O(k)$ and estimate the $k' \times k'$ principal submatrix of $T$ indexed by $\{0, a, \ldots, k'a \mod d\}$, we are equivalently measuring the top-left $k' \times k'$ submatrix of a \emph{transformed} matrix $\tilde{T}$ whose Vandermonde decomposition frequencies are $\{\tilde{f}_1, \ldots, \tilde{f}_d\}$ where $\tilde{f}_j  = a\cdot {f}_j \mod 1$. Ideally, we would estimate the frequencies of $\tilde{T}$, which are separated by larger gaps, and use them to recover the frequencies of $T$. However, this cannot be done directly because there is \emph{ambiguity} in inverting each $\tilde{f}_j$: there are up to $a$ different solutions $f_j \in [0,1]$ to the equation ${f}_j \equiv a\tilde{f}_j \mod 1$, as shown in Figure \ref{fig:ambiguity}.

Fortunately, this issue can be combated with simple repetition. Each time we draw a different random $a$, we collect potential candidate dominant frequencies for $T$'s Vandermonde decomposition. Since we restrict $a \le bd/k$, there will be $a \cdot k \leq b \cdot d$ such candidates: $k$ will be the true dominant frequencies in $T$'s Vandermonde decomposition and the remainder will be nearly random. Roughly, any frequency outside the set of dominant frequencies  will appear in the set with probability $b < 1$. Thus, setting $b$ small enough, after roughly $O(\log d)$ repetitions, by observing which frequencies appear as candidates the largest number of times, we can determine the true dominant $k$ frequencies with high probability. 

\begin{figure}[h]
\centering
\includegraphics[width=0.4 \linewidth]{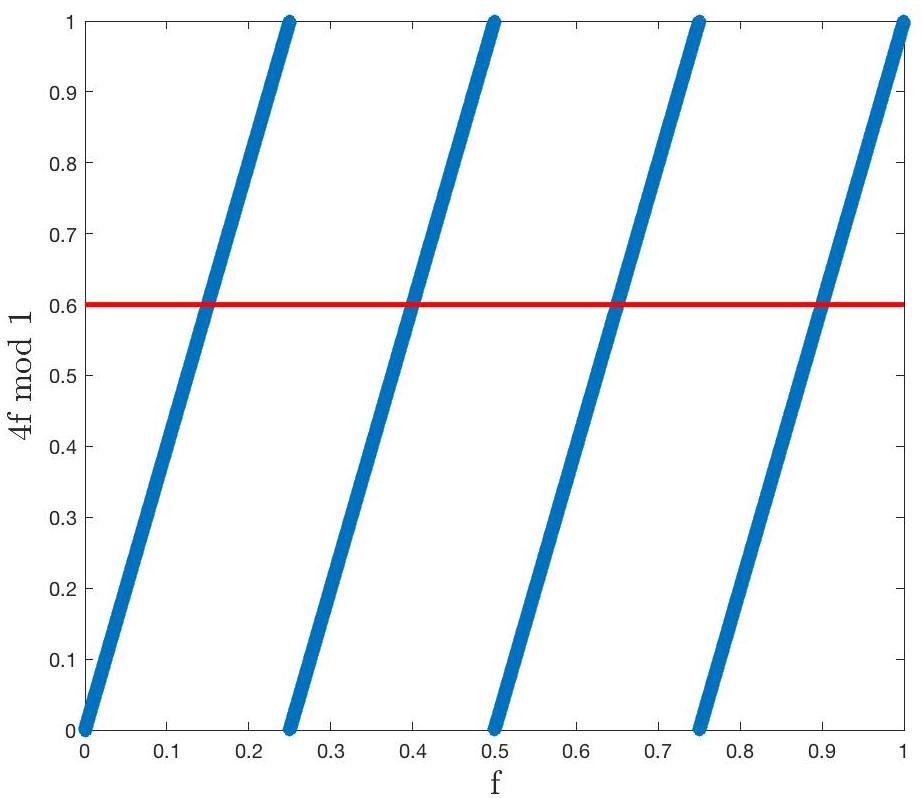}
\caption{An example of the ambiguity induced by permuting, when frequencies are off-grid. Here, we suppose $a=4$ and the recovered (post-permutation) frequency is $\tilde{f}_j=0.6$. As shown, there are $4$ possible ``true" frequencies that may have given rise to $\tilde{f}_j=0.6$. }
\label{fig:ambiguity}

\end{figure}

\section{Analysis for Circulant Covariance}
\label{sec:analysis}

We now apply the random ultra-sparse ruler construction of Definition \ref{def:rus} to circulant covariance matrix estimation. For the remainder of the section let $F \in \C^{d \times d}$ be  the discrete Fourier transform matrix with $F(j,k) = \frac{1}{\sqrt{d}} \cdot  e^{\frac{2\pi i (j-1)(k-1)}{d}}$. For $x \in \R^d$ let $F^* x  = \hat x$ denote its Fourier transform. Let $\diag(x)$ be the diagonal matrix with $x$ on its diagonal, and let $\toep(x)$ be the symmetric Toeplitz matrices with first column $x$. Our algorithm makes blackbox use of a random hashing based sparse Fourier transform (SFT, or sparse FFT), with output guarantees as follows:
\begin{theorem}[Sparse Fourier Transform \cite{sft}]\label{thm:sft} Consider $x \in \R^d$ with  Fourier transform $\hat{x}$.  Assume that  $d$ is a power of $2$. 
 Let $\delta>0$ be a fixed error parameter. Algorithm 4.1 of \cite{sft} $SFT(x)$ 
 outputs $k$-sparse $\hat{z}$ satisfying with $\frac{2}{3}$ probability:
$$ ||\hat{x}-\hat{z}||_2 \leq 2 \min_{\text{k-sparse }y}|| \hat{x}-y||_2 + \delta||\hat{x}||_2.$$
The algorithm reads $O(\log^2 \frac{d}{k})$ blocks of $O\left(k \log \frac{d}{\delta} \right)$ entries. Each block consists of the first entries of $x$ after applying a different random permutation $g_{a,c}(\cdot)$ for $a,c$ chosen uniformly from $\{1,\ldots,d\}$ with $a$ odd (and thus coprime to $d$). 
\end{theorem}

By Lemma \ref{clm:cyclic}, the sparse Fourier transform algorithm of Theorem \ref{thm:sft} can be implemented via random ultra-sparse rulers with low  entry sample complexity in the covariance estimation setting:
\begin{corollary}\label{cor:cyclic}
There is a ruler $R$ with $O\left (\sqrt{k \log\frac{d}{\delta}} \cdot \log^2\frac{d}{k}\right )$ elements measuring all distances required for the algorithm of Theorem \ref{thm:sft} to be applied to the first column of any circulant matrix $T$.
\end{corollary}
\begin{proof}
Let $R$ be the union of  random ultra-sparse rulers $R_{a,c}$, each of which measures the entries in a permuted block read by  the algorithm of Theorem \ref{thm:sft} \cite{sft}. By  Lemma \ref{clm:cyclic}, each $R_{a,c}$ is guaranteed to exist with just $O(\sqrt{k \log \frac{d}{\delta}})$ entries and there are $O(\log^2\frac{d}{k})$ blocks. Thus $|R| \leq O\left (\sqrt{k \log\frac{d}{\delta}} \cdot \log^2\frac{d}{k}\right )$.
\end{proof}
With Corollary \ref{cor:cyclic} in place, we present our main algorithm (Algorithm \ref{alg:alg}). Note that in this algorithm, $\bar t_s$ is only  estimated at the $\tilde O(k)$ positions represented by the ruler $R$ (i.e. in the difference set of $R$). Since $SFT(\bar  t)$ only requires reading $\bar t$ at these positions, its output \emph{does not depend on the other positions}.
We have:
\begin{lemma}\label{lem:5}
Consider circulant covariance matrix $T \in \R^{d \times d}$ for a sub-Gaussian distribution $\mathcal{D}$.
Let $t \in \R^d$ be the first  column of $T$ and $\bar t \in \R^d$ be the estimate computed by Algorithm \ref{alg:alg} (line 2). Let $w \in \R^d$ match $\bar t$ on all entries read by SFT and match $t$ elsewhere. Letting $m = O \left (k \cdot \log^2 \frac{d}{k} \cdot \log\frac{d}{\epsilon} \right)$ (then number of entries of $\bar t$ that are read by SFT in line 3), and $n = O \left (\frac{m \sqrt{\log m}}{\epsilon^2} \right )$, we have with probability at least $2/3$, $\norm{t - w}_2 \le \epsilon \norm{T}_2$ and further, Algorithm \ref{alg:alg} outputs $\hat z$ with:
\begin{align*}
\norm{\hat z - 
\hat t}_2 \le 5 \epsilon \cdot  t_0 + 2 \min_{\text{k-sparse }y}|| \hat{t}-y||_2.
\end{align*}
Note that Algorithm \ref{alg:alg} has entry sample complexity $\tilde O(\sqrt{k \log \frac{d}{\epsilon}} \cdot \log^2 \frac{d}{k})$ (see Corollary \ref{cor:cyclic}) and vector sample complexity $n$.
\end{lemma}

 \begin{algorithm}[h]
\caption{{Covariance estimation via ultra-sparse ruler}}
{\bf input}: i.i.d. samples $x^{(1)},\ldots,x^{(n)} \sim \mathcal{D}$ with Toeplitz covariance $T$. Random ultra-sparse ruler $R$ from Corollary \ref{cor:cyclic} with parameters $d,k$, and $\delta = \frac{\epsilon}{\sqrt{k}}$.\\
{\bf output}: $\tilde T \in \R^{d \times d}$ approximating $T$.
\begin{algorithmic}[1]
\State{Let $R(s) := \{(i,j) \in R \text{ s.t. } i-j = s \text{ or }i-j = d-s\}$ .}
\State{Let $\bar t \in \R^{d}$ be given by: for $s$ measured by $R$, let ${\bar t}_s := \frac{1}{n |R(s)|}\cdot\sum_{\ell=1}^n \sum_{(i,j) \in R(s)} x^{(\ell)}_i x^{(\ell)}_j$. Let ${\bar t}_s = 0$ otherwise.}
\State{Compute $\hat z  := SFT(\bar t)$ where $\bar t = [\bar t_0,\ldots, \bar t_{d-1}]$.}
\State{Let $z := F \hat z$.}
\\\Return{$\tilde {T} = \toep(z)$.}
\end{algorithmic}
\label{alg:alg}
\end{algorithm}

\begin{proof}
We  apply  Theorem \ref{thm:sft} to input $w$ with $\delta = \epsilon/\sqrt{d}$,  which requires reading $m =O \left (k \cdot \log^2(\frac{d}{k}) \cdot \log(\frac{d}{\epsilon} )\right) $ entries of $w$.
Since $\mathcal{D}$ is sub-Gaussian and since $t_0$ is the largest entry to $T$ by positive semidefiniteness, for $n = O \left (\frac{m \sqrt{\log m}}{\epsilon^2} \right )$, we have $|t_s - \bar t_s| \le \frac{\epsilon t_0}{\sqrt{m}}$ for each $s$ measured by $R$ with probability $1/\Theta(m)$. By a union bound, this approximation then holds for all $s$ with good probability. We thus have $\norm{w - t}_2 \le \sqrt{m \cdot  \frac{\epsilon^2 t_0^2}{m}} = \epsilon t_0$. By Parseval's theorem, it follows that $\norm{\hat w - \hat t}_2 \le  \epsilon t_0$, and by the triangle inequality: $$\min_{\text{k-sparse }y}|| \hat{w}-y||_2 \le ||\hat{w}-\hat{t}|| + \min_{\text{k-sparse }y}|| \hat{t}-y||_2 \le \epsilon t_0 + \min_{\text{k-sparse }y}|| \hat{t}-y||_2.$$
 By Theorem  \ref{thm:sft} with $\delta = \epsilon/\sqrt{d}$:
\begin{align*}
\norm{\hat z - 
\hat w}_2  \le 2\min_{\text{k-sparse }y}||\hat{w}-y|| + \delta ||\hat{w}|| \leq 2\epsilon t_0 + 2 \min_{\text{k-sparse }y}|| \hat{t}-y||_2 + \frac{\epsilon}{\sqrt{d}} (\norm{\hat t}_2 + \epsilon t_0).
\end{align*}
Again noting that by positive semidefiniteness, $t_0$ is the largest entry in $t$, we have $\frac{\epsilon}{\sqrt{d}} \norm{\hat t}_2 =  \frac{\epsilon}{\sqrt{d}} \norm{t}_2 \le \epsilon t_0$. This gives $\norm{\hat z - \hat w}_2 \le 4 \epsilon t_0 + 2 \min_{\text{k-sparse }y}|| \hat{t}-y||_2$. The claim follows  by  applying the triangle inequality one more time to bound $\norm{\hat z - \hat t}_2 \le \norm{\hat z - \hat w}_2 + \norm{\hat w - \hat t}_2 \le \norm{\hat z - \hat w}_2 + \epsilon t_0$.
\end{proof}
Finally, we prove Theorem \ref{thm:main} by using the above bound on $\norm{\hat z - 
\hat t}_2$ to bound $\norm{T - \tilde T}_F = \norm{\toep(t)-\toep(z)}_F$.

\begin{lemma}
If the bound of Lemma \ref{lem:5} holds:
$
\norm{T - \tilde T}_F\le 5\epsilon \norm{T}_F + 2 \min_{\rank-k\ B} \norm{T - B}_F
$
\end{lemma}
\begin{proof}
Both $\toep(t)$ and $\toep(z)$ are circulant, and so can be written in their eigendecompositions as $F D F^*$ where $D = \sqrt{d} \diag(F^* t) = \sqrt{d} \diag(\hat t)$ and $\tilde D = \sqrt{d} \diag(F^* z) = \sqrt{d}\diag(\hat z)$. Thus:
\begin{align*}
\norm{\toep(t) - \toep(z)}_F  &=  \norm{FDF^* - F\tilde{D}F^*}_F\\
&= \norm{ D - \tilde D}_F \\
&= \sqrt{d}\norm{\diag(F^*t)-\diag(F^*z)}\\
&= \sqrt{d}\norm{\hat t- \hat z}_2 \le \sqrt{d} \cdot 5 \epsilon \cdot  t_0 + \sqrt{d} \cdot 2 \min_{\text{k-sparse }y}|| \hat{t}-y||_2,
\end{align*}
where the last bound follows from Lemma \ref{lem:5}.
We can see that $\sqrt{d} \cdot 5\epsilon t_0 \le 5\epsilon \norm{T}_F$.  Further, the  best rank-$k$ approximation of $T$ is given by projecting onto its top $k$-eigenvectors (equivalently, setting to zero all but the largest $k$ entries of $D$ to obtain $D_k$, or approximating $\hat t$ with its best $k$-sparse approximation, $t_k$). We thus have:

\begin{align*}
\sqrt{d} \cdot \min_{\text{k-sparse }y}|| \hat{t}-y||_2 = \sqrt{d} \cdot ||\hat{t}-\hat{t_k}|| = \norm{D-D_k} = \norm{FDF^*-FD_kF^*} = \min_{\rank-k\ B} \norm{T - B}_F,
\end{align*}
which yields 
$$\norm{\toep(t) - \toep(z)}_F   \le 5\epsilon \norm{T}_F + 2 \min_{\rank-k\ B} \norm{T - B}_F,$$
completing the proof.
\end{proof}

%

%

\section{Experimental Validation}
\label{sec:experiments}

We conclude by experimentally evaluating the driving intuition behind random ultra-sparse rulers: when there is a small frequency gap in $T$'s Vandermonde decomposition, it can be very advantageous to randomly permute the frequencies to remove this gap. To do so, we generate a low-rank, positive semidefinite real Toeplitz matrix with on-grid but clustered frequencies, add entrywise noise $\eta \sim N(0,\nu)$ and apply the following simple reconstruction procedure. Given a subset of noise-corrupted measurements of $T$'s first column $t$, we use the \texttt{pmusic} and \texttt{findpeaks} functions in Matlab to identify $k$ estimated frequencies $\tilde{f}_1,\dots,\tilde{f}_{k}$, and solve the appropriate linear regression problem to recover diagonal $\tilde D$ so that $T$ is approximated by $\tilde T = \tilde{F}\tilde D\tilde{F}^*$, where $\tilde{F}$ is the $n \times k$ Fourier matrix corresponding to $\tilde{f}_1,\dots,\tilde{f}_k$. We note that this simple reconstruction approach matches that of \cite{qiao2017gridless} \emph{up to a preliminary denoising step}. This step could be applied to all sampling schemes and should preserve their relative performance. The sampling schemes compared are: 
\begin{enumerate}
 \item \emph{First $O(k)$ samples}: Input to \texttt{pmusic} $4k$ noisy estimates of $t_0,\dots,t_{4k-1}$, which can be measured from samples $x^{(i)} \sim \mathcal{D}$ via existing ruler constructions with entrywise sample complexity $O(\sqrt{k})$. This corresponds to the approach of \cite{qiao2017gridless}.
 
 \item \emph{Permuted $O(k)$ samples}: Input to \texttt{pmusic} $4k$ noisy samples of $t_{g^{-1}_{a,c}(0)},\ldots,t_{g^{-1}_{a,c}(4k-1)}$, where $g_{a,c}$ is as described in Section \ref{sec:contributions}. For simplicity we take $c = 0$.
These samples can be obtained with entrywise sample complexity $\tilde{O}(\sqrt{k})$ using a random ultra-sparse ruler by Lemma \ref{clm:cyclic}. 

\item \emph{All samples}: As a baseline, input to \texttt{pmusic} all $d$ noisy measurements of $t_0,\dots,t_{d-1}$. 
\end{enumerate}

Experimental results and validation are shown in Figure 2. As expected, sampling scheme (3) (which requires $O(\sqrt{d})$ entrywise sample complexity) performs best, but is closely followed by our proposed permutation-based sampling method. More elaborate reconstructions following the full algorithm of \cite{sft} would likely improve further on this simple algorithm. Nonetheless, it is clear that when T is circulant with some frequencies clustered, the permutation approach enabled by random ultra-sparse rulers can vastly improve robustness to noise while retaining low, $\tilde{O}(\sqrt{k})$ entrywise sample complexity.


\begin{figure}[htb]
\label{fig:results}
\begin{minipage}[b]{1.0\linewidth}
  \centering
  \centerline{\includegraphics[width=10cm]{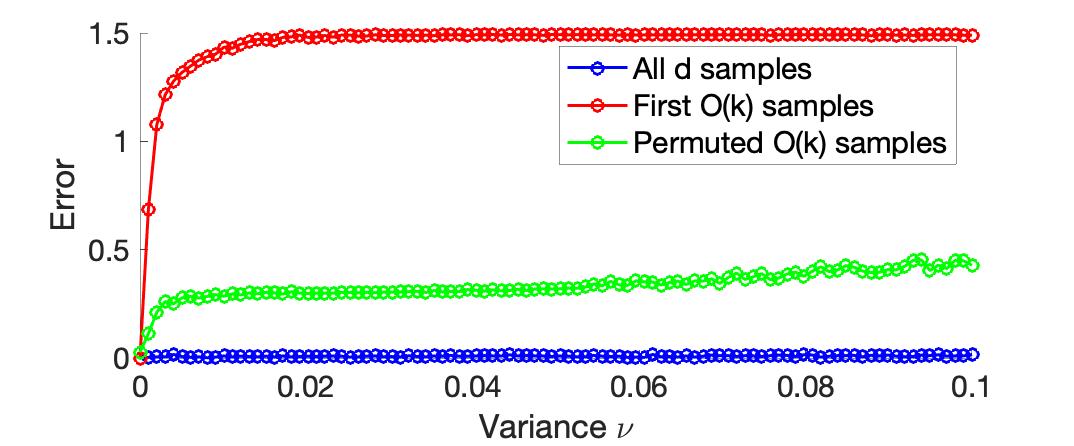}}
  \centerline{(a) Estimation error as a function of noise variance.}\medskip
\end{minipage}
\begin{minipage}[b]{.49\linewidth}
  \centering
  \centerline{\includegraphics[width=7cm]{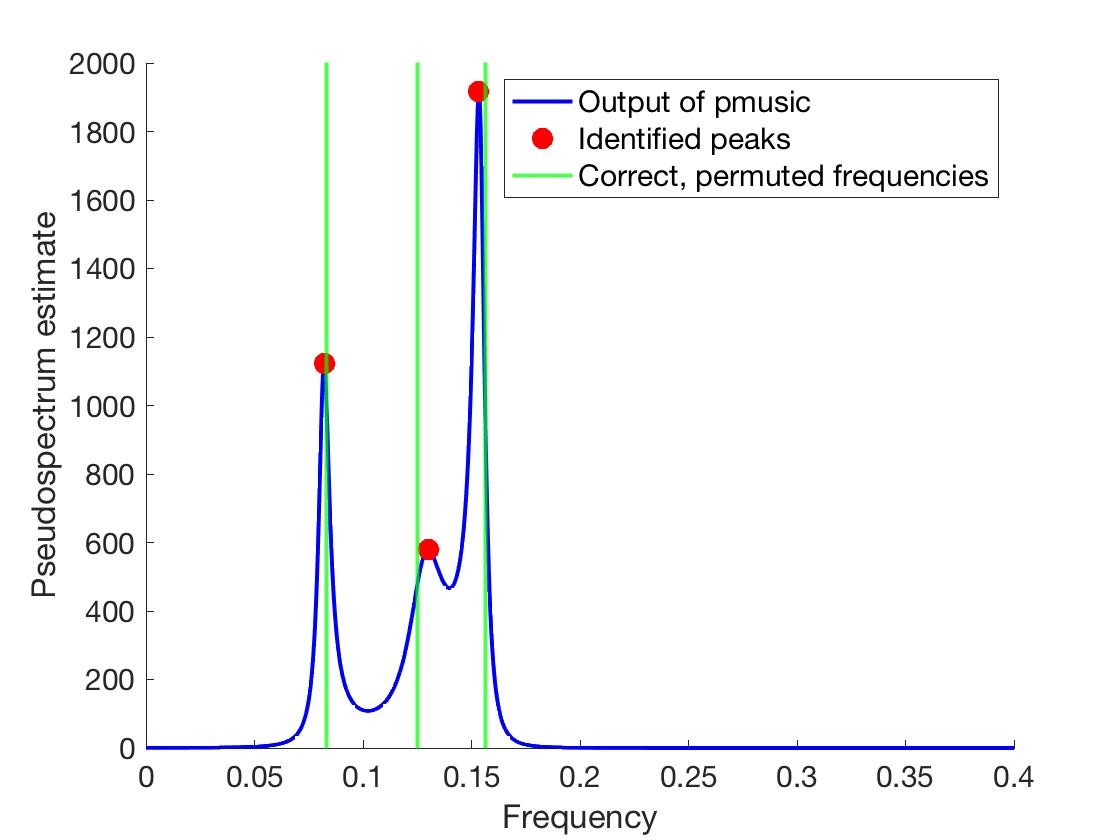}}
  \centerline{(b) With permutation}\medskip
\end{minipage}
\hfill
\begin{minipage}[b]{0.49\linewidth}
  \centering
  \centerline{\includegraphics[width=7cm]{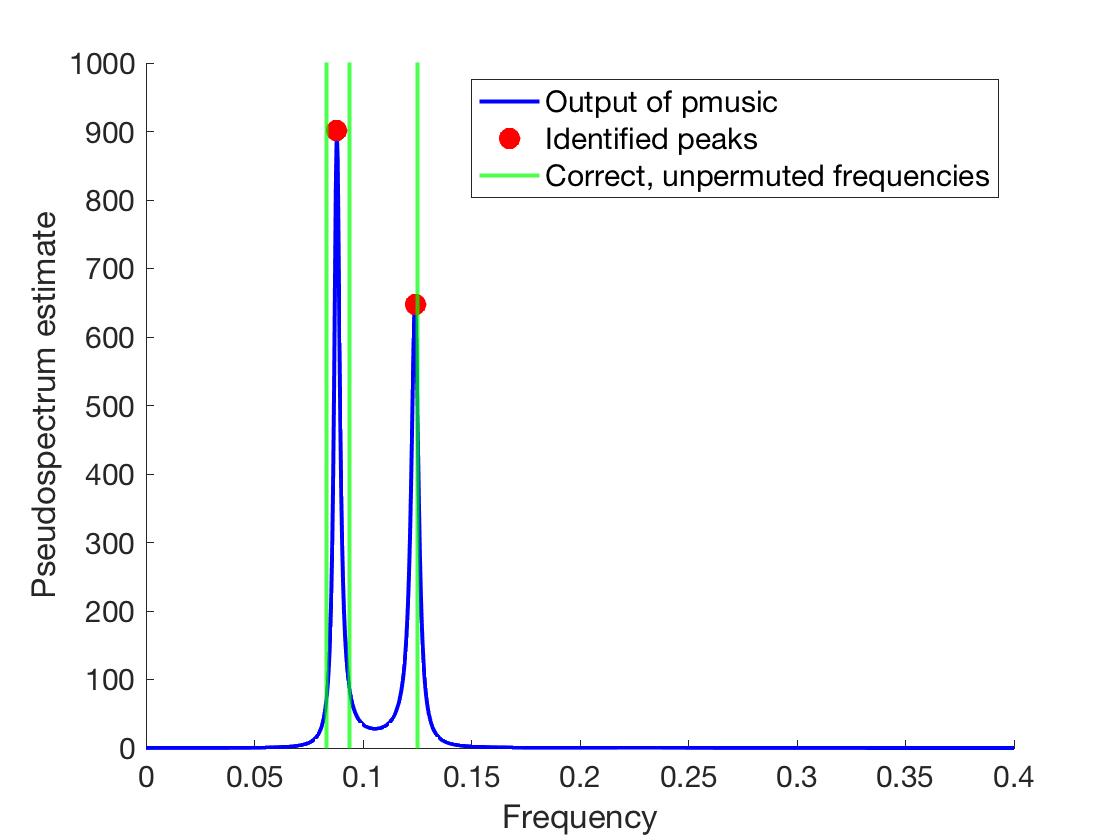}}
  \centerline{(c) Without permutation}\medskip
\end{minipage}
\caption{Normalized estimation error $\frac{||t-\tilde{t}||_2}{||t||_2}$ of different ruler-enabled sampling schemes as a function of the  variance $\nu$. Here, $k=6$ and $d=2400$, with minimum frequency gap $\approxeq 0.01 < \frac{1}{k}$. Results are averaged over 10 random permutations, each of which is further averaged over 20 trials.  In subfigures (b) and (c), we demonstrate for a single iteration at $\nu=0.5$ how nearby frequencies are conflated without permutation, but likely to be separated and accurately identified with a permutation. As the frequencies are symmetric (to ensure T is real), only the first $\frac{k}{2}=3$ are shown.}
\label{fig:res}
\end{figure}

\subsection*{Acknowledgements}

We thank Yonina Eldar for many valuable conversations on sparse rulers and covariance estimation. We also thank Piya Pal for bringing the related  work of \cite{qiao2017gridless} to our attention.


\vfill\pagebreak

\bibliographystyle{IEEEbib}
\bibliography{ICASSPsqrtk}

\end{document}